\documentclass[11pt]{article}

\usepackage{latexsym}
\usepackage{epsfig}
\usepackage{amsmath}
\usepackage{amsfonts}
\usepackage{amssymb}
\usepackage{wrapfig}
\usepackage{graphics}
\usepackage{url}

\newtheorem{theorem}{Theorem}
\newtheorem{lemma}{Lemma}
\newtheorem{proposition}{Proposition}
\newtheorem{corollary}{Corollary}
\newtheorem{assumption}{Assumption}

\newtheorem{definition}{Definition}
\newtheorem{conjecture}{Conjecture}

\newenvironment{proof}
  {\bigskip\noindent\textbf{Proof~}\itshape}
  {\marginpar{$\Box$}\bigskip}

\newcommand{\qed}
{\nobreak \ifvmode \relax \else
 \ifdim\lastskip<1.5em \hskip-\lastskip
 \hskip1.5em plus0em minus0.5em \fi}

\newcommand{\ie}{{i.e.}}
\newcommand{\eg}{{e.g.}}

\begin{document}

\title{Network Tomography Based on Additive Metrics}

\author{Jian Ni \hspace{0.1in} Sekhar Tatikonda \\ Department of Electrical Engineering, Yale University, New Haven, CT, USA}

\maketitle

\begin{abstract}

Inference of the network structure (e.g., routing topology) and dynamics (e.g., link performance)
is an essential component in many network design and management tasks. In this paper we propose a
new, general framework for analyzing and designing routing topology and link performance inference
algorithms using ideas and tools from phylogenetic inference in evolutionary biology. The framework
is applicable to a variety of measurement techniques. Based on the framework we introduce and
develop several polynomial-time distance-based inference algorithms with provable performance. We
provide sufficient conditions for the correctness of the algorithms. We show that the algorithms
are consistent (return correct topology and link performance with an increasing sample size) and
robust (can tolerate a certain level of measurement errors). In addition, we establish certain
optimality properties of the algorithms (i.e., they achieve the optimal $l_\infty$-radius) and
demonstrate their effectiveness via model simulation.

\end{abstract}

\begin{index}
Network tomography, routing topology inference, link performance estimation, additive metrics,
neighbor-joining.
\end{index}

\section{Introduction}

Network tomography (network inference) \cite{NetworkTomography,InternetTomography,Vardi} is an
emerging field in communication networks which studies the estimation and inference of the network
structure and dynamics (\eg, routing topology, link performance, traffic demands) based on
\emph{indirect} measurements when \emph{direct} measurements are unavailable or difficult to
collect. As modern communication networks (\eg, the Internet, wireless communication networks)
continue to grow in size, complexity, and diversity, scalable and accurate network inference
algorithms and tools will become increasingly important for many network design and management
tasks. These include service provision and resource allocation, traffic engineering, network
monitoring, application design, etc.

In \emph{network monitoring}, such tools can help a network operator obtain routing information and
network internal characteristics (\eg, loss rate, delay, utilization) from its network to a set of
other collaborating networks that are separated by non-participating autonomous networks. If the
performance of a certain portion of the network experiences sudden, dramatic changes, it can be an
indication of failures or anomalies occurred in that portion of the network.

In \emph{application design}, such tools can be particularly useful for peer-to-peer (P2P) style
applications where a node communicates with a set of other nodes (called \emph{peers}) for file
sharing and multimedia streaming. For example, a node may want to know the routing topology to
other nodes so that it can select peers with low or no route overlap to improve resilience against
network failures (\eg, \cite{RON}). As another example, a streaming node using multi-path may want
to know both the routing topology and link loss rates so that the selected paths have low loss
correlation (\eg, \cite{AKMS}).

So far there are two primary approaches to infer the routing topology and link performance of a
communication network. An \emph{internal-assisted} approach uses tools based on measurements or
feedback messages of the internal nodes (\eg, routers). Such an approach is limited as today's
communication networks are evolving towards more decentralized and private adminstration. For
example, a common approach to infer the routing topology in the Internet is to use
\emph{traceroute}. However, an increasing number of routers in the Internet will block traceroute
requests due to privacy and security concerns. These routers are known as \emph{anonymous routers}
\cite{AnonymousRouters} and their existence makes the routing topology inferred by traceroute
inaccurate. In addition, administrators of different networks normally will not reveal or share
their link-level measurement data for us (\eg, end hosts) to derive the link performance.

Not depending on extra cooperation from the internal nodes (except the basic packet forwarding
functionality), a \emph{network tomography} approach utilizes end-to-end packet probing
measurements (such as packet loss and delay measurements) conducted by the end hosts to infer the
routing topology and link performance. Under a network tomography approach, a source node will send
probes to a set of destination nodes. The basic idea is to utilize the correlations among the
observed losses and delays of the probes at the destination nodes to infer the routing topology and
link performance from the source node to the destination nodes. Due to its flexibility and
reliability, network tomography has attracted many recent studies. Both multicast probing based
approaches (\eg, \cite{MulticastLoss}, \cite{MulticastTopology}, \cite{NTisit06},
\cite{MulticastDelay}, \cite{Grouping}) and unicast probing based approaches (\eg,
\cite{UnicastLoss}, \cite{UnicastTopology}, \cite{StripedUnicast}, \cite{HierarchicalUnicast},
\cite{UnicastDelay}) have been investigated.

The main challenges of existing approaches and techniques include
\begin{itemize}
\item \textbf{computational complexity};

\item \textbf{information fusion}: how to fuse information from different measurements to achieve
the best estimation accuracy;

\item \textbf{probing scalability} (especially under unicast probing);

\item \textbf{node dynamics}: how to handle dynamic node joining and leaving efficiently.
\end{itemize}

In this paper we propose a new, general framework for designing and analyzing topology and link
performance inference algorithms using ideas and tools from phylogenetic inference in evolutionary
biology. The framework is built upon additive metrics. Under an additive metric the path metric
(path length) is expressed as the summation of the link metrics (link lengths) along the path. The
basic idea is to use (estimated) distances between the terminal nodes (end hosts) to infer the
routing tree topology and link metrics. Based on the framework we introduce and develop several
computationally efficient inference algorithms with provable performance.

The advantages of our framework are summarized as follows.

\begin{itemize}

\item The framework is applicable to a variety of measurement techniques, including multicast
probing, unicast probing, and traceroute probing. Since a linear combination of different additive
metrics is still an additive metric, the framework can flexibly fuse information available from
different measurements to achieve better estimation accuracy.

\item Based on the framework we can design, analyze, and develop distance-based inference
algorithms that are \emph{computationally efficient} (polynomial-time), \emph{consistent} (return
correct topology and link performance with an increasing sample size), and \emph{robust} (can
tolerate a certain level of measurement errors).

\end{itemize}

We organize the paper as follows. In Section \ref{sec:NTmodel} we describe the network model and
the inference problem. In Section \ref{sec:NTadditive} we introduce additive metrics on trees, and
we discuss how to construct additive metrics and compute/estimate the distances between the
terminal nodes from end-to-end measurements. In Section \ref{sec:NTNJ} we introduce the
neighbor-joining (NJ) algorithm for constructing binary trees from distances. In Section
\ref{sec:NTRNJ} we propose a rooted version of the NJ algorithm and extend it to general trees. In
Section \ref{sec:NTsimulation} we demonstrate the effectiveness of the inference algorithms via
model simulation. In Section \ref{sec:NTMSSD} we extend our framework to infer the routing topology
and link performance from multiple source nodes to a single destination node. We summarize the
paper in Section \ref{sec:NTsummary}.

\section{Network Model and Inference Problem} \label{sec:NTmodel}

Let $\mathcal{G}=(\mathcal{V},\mathcal{E})$ denote the topology of the network, which is a directed
graph with node set $\mathcal{V}$ (end hosts, internal switches and routers, etc.) and link set
$\mathcal{E}$ (communication links that join the nodes). For any nodes $i$ and $j$ in the network,
if the underlying routing algorithm returns a sequence of links that connect $j$ to $i$, we say $j$
is \emph{reachable} from $i$. We call this sequence of links a \emph{path} from $i$ to $j$, denoted
by $\mathcal{P}(i,j)$. We assume that during the measurement period, the underlying routing
algorithm determines a unique path from a node to another node that is reachable from it.

Hence the \emph{physical routing topology} from a source node to a set of destination nodes is a
(directed) tree. From the physical routing topology, we can derive a \emph{logical routing tree}
which consists of the source node, the destination nodes, and the \emph{branching nodes} (internal
nodes with at least two outgoing links) of the physical routing tree (\eg, \cite{MulticastLoss},
\cite{MulticastTopology}, \cite{Grouping}). Notice that a logical link may comprise more than one
consecutive physical links, and the degree of an internal node on the logical routing tree is at
least three. An example is shown in Fig. \ref{fig:physical}. For simplicity we use \emph{routing
tree} to express logical routing tree unless otherwise noted.

\begin{figure}[t]
\center{\psfig{file=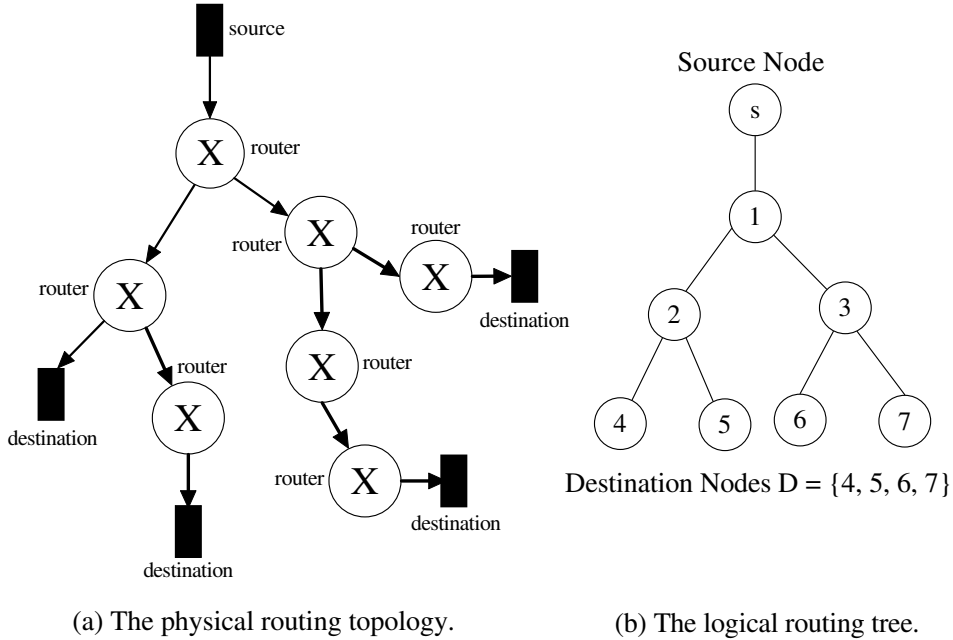, width=\columnwidth}} \caption{The physical routing
topology and the associated logical routing tree with a single source node and multiple destination
nodes.} \label{fig:physical}
\end{figure}

Suppose $s$ is a source node in the network, and $D$ is a set of destination nodes that are
reachable from $s$. Let $T(s,D)=(V,E)$ denote the routing tree from $s$ to nodes in $D$, with node
set $V$ and link set $E$. Let $U=s\cup D$ be the set of terminal nodes, which are nodes of degree
one (\eg, end hosts).

Every node $k\in V$ has a \emph{parent} $f(k)\in V$ such that $(f(k),k) \in E$, and a set of
\emph{children} $c(k)=\{j \in V: f(j)=k\}$, except that the source node (root of the tree) has no
parent and the destination nodes (leaves of the tree) have no children. For notational
simplification, we use $e_k$ to denote link $(f(k),k)$.

Each link $e\in E$ is associated with a performance parameter $\theta_e$ (\eg, success rate, delay
distribution, utilization, etc.). The \textbf{network inference problem} involves using
measurements taken at the terminal nodes to infer
\begin{itemize}
\item[(1)] the topology of the (logical) routing tree;

\item[(2)] link performance parameters $\theta_{e}$ of the links on the routing tree.
\end{itemize}

We want to point out that the network inference problem is similar to the \emph{phylogenetic
inference problem} in evolutionary biology. The phylogenetic inference problem is to determine the
evolutionary relationship among a set of species. Such relationship is often represented by a
phylogenetic tree, in which the terminal nodes represent extant species and the internal nodes
represent extinct common ancestors of the extant species. Many methods have been developed to
reconstruct phylogenetic trees from biological information (\eg, biomolecular sequence data)
observed at the terminal nodes (\eg, \cite{InferPhylogenies}, \cite{Phylogenetics}). The
mathematical models of the two problems are very similar, except that in the network inference
problem we can control and observe the source node, while in the phylogenetic inference problem the
information of the source node (the common ancestor of all species) is lost. We will use ideas and
tools from phylogenetic inference to analyze and solve the network inference problem.

\section{Additive Metrics on Trees}   \label{sec:NTadditive}

The tool that we will use to analyze and solve the network inference problem is the so-called
\emph{additive tree metric} \cite{Phylogenetics}, or \emph{additive metric} for short. We consider
trees with internal node degree at least three. Notice that all (logical) routing trees have such
property.

\begin{definition}
$d:V\times V\rightarrow \mathbb{R}^+$ is an \textbf{\emph{additive metric}} on $T=(V,E)$ if
\begin{eqnarray*}
\begin{split}
(a)\phantom{aa} &  0 < d(e) < \infty, \phantom{aa}\forall  e=(i,j)\in E; \\
(b)\phantom{aa} &  d(i,j) = d(j,i) = \left\{
\begin{array}{ll}
\sum_{e \in \mathcal{P}(i,j)}d(e), & i\neq j; \\
0, &  i=j.
\end{array}\right.
\end{split}
\end{eqnarray*}
\end{definition}

$d(e)$ can be viewed as the \emph{length} of link $e$, and $d(i,j)$ can be viewed as the
\emph{distance} between nodes $i$ and $j$. Basically, an additive metric associates each link on
the tree with a finite positive link length, and the distance between two nodes on the tree is the
summation of the link lengths along the path that connects the two nodes.

Suppose $T(s,D)=(V,E)$ is a routing tree with source node $s$ and destination nodes $D$. Let
\begin{eqnarray*}
d(E) & = & \Big\{d(e): e\in E\Big\}
\end{eqnarray*}
denote the link lengths of $T(s,D)$ under additive metric $d$.

Remember $U=s\cup D$ is the set of terminal nodes on the tree. Let
\begin{eqnarray*}
d(U^2) & = & \Big\{d(i,j): i, j\in U\Big\}
\end{eqnarray*}
denote the distances between the terminal nodes.

Buneman \cite{Buneman} showed that the topology and link lengths of a tree are uniquely determined
by the distances between the terminal nodes under an additive metric.

\begin{theorem} \label{theorem:Buneman}
There is a one-to-one mapping between $(T(s,D),d(E))$ and $(U,d(U^2))$ under any additive metric
$d$ on $T(s,D)$.
\end{theorem}

From Theorem \ref{theorem:Buneman}, we know that we can recover the topology and link lengths of a
routing tree if we know $d(U^2)$. In addition, if there is a one-to-one mapping between the link
performance parameters and link lengths (which will be clear in Section \ref{sec:NTmulticast}),
then we can recover the link performance parameters from the link lengths. The challenges are:
\begin{itemize}
\item[(1)] Constructing an additive metric for which we can derive/estimate $d(U^2)$ from
measurements taken at the terminal nodes. We will address this issue in this section.

\item[(2)] Developing efficient and effective algorithms to recover the topology and link lengths
from the (estimated) distances between the terminal nodes. We will address this issue in Sections
\ref{sec:NTNJ}, \ref{sec:NTRNJ}.
\end{itemize}

\subsection{Construct Additive Metrics} \label{sec:NTmulticast}

A source node can employ different probing techniques, \eg, \emph{multicast} probing and
\emph{unicast} probing, to send probes (packets) to a set of destination nodes. For multicast
probing, when an internal node on the routing tree receives an packet from its parent, it will
duplicate the packet and send a copy of the packet to all its children on the tree. Therefore, the
packets received by different destination nodes have exactly the same network experience (loss,
delay, etc.) in the shared links.

For a (multicast) probe sent by source node $s$ to the destination nodes in $D$, we define a set of
\emph{link state variables} $Z_e$ for all links $e\in E$ on the routing tree $T(s,D)$. $Z_e$ takes
value in a state set $\mathcal{Z}$. The distribution of $Z_e$ is parameterized by $\theta_{e}$,
\eg,
\begin{eqnarray}
\mathbb{P}(Z_e=z) & = & \theta_e(z), \phantom{aa} \forall z\in \mathcal{Z}.
\end{eqnarray}

The transmission of a probe from $s$ to nodes in $D$ will induce a set of \emph{outcome variables}
on the routing tree. For each node $k\in V$, we use $X_{k}$ to denote the (random) outcome of the
probe at node $k$. $X_k$ takes value in an outcome set $\mathcal{X}$. By \emph{causality} the
outcome of the probe at node $k$ (\ie, $X_k$) is determined by the outcome of the probe at node
$k$'s parent $f(k)$ (\ie, $X_{f(k)}$) and the link state of $e_k$ (\ie, $Z_{e_k}$):
\begin{eqnarray}\label{eq:outcome}
X_k & = & g(X_{f(k)}, Z_{e_k}).
\end{eqnarray}

\begin{assumption} \label{assump:independence}
The link states are independent from link to link (spatial independence assumption) and are
stationary during the measurement period (stationarity assumption).
\end{assumption}

\begin{proposition} \label{prop:MRF}
Under the spatial independence assumption that the link states are independent from link to link,
\begin{eqnarray}
X_V & \stackrel{\Delta}{=} & (X_{k}: k\in V)
\end{eqnarray}
is a Markov random field (MRF) on $T(s,D)$. Specifically, for each node $k\in V$, the conditional
distribution of $X_{k}$ given other random variables $(X_{j}: j\neq k)$ on $T(s,D)$ is the same as
the conditional distribution of $X_{k}$ given just its neighboring random variables $(X_{j}: j \in
f(k)\cup c(k))$ on $T(s,D)$.
\end{proposition}

\begin{proof}
For notational simplification, we use $p(x_A)$ to represent $\mathbb{P}(X_k=x_k:k\in A)$ for any
subset $A \subseteq V$. First we prove by induction that
\begin{eqnarray}\label{eq:tree}
p(x_V) & = & p(x_s)\prod_{k\in V\setminus s}p(x_k|x_{f(k)}).
\end{eqnarray}
Equation (\ref{eq:tree}) is clearly true for any tree with $|V|=1$ or $|V|=2$. Assume
(\ref{eq:tree}) is true for any tree with $|V|\leq n$. Now consider a tree $T$ with $|V|=n+1$.

Let $i$ be a leaf node of $T$, then by (\ref{eq:outcome}) and the spatial independence assumption
we have
\begin{eqnarray}\label{eq:induction}
p(x_V) & = & p(x_i|x_{V\setminus i})p(x_{V\setminus i}) \nonumber \\
       & = & p(g(x_{f(i)},z_{e_i})|x_{V\setminus i})p(x_{V\setminus i}) \nonumber\\
       & = & p(g(x_{f(i)},z_{e_i})|x_{f(i)})p(x_{V\setminus i}) \nonumber\\
       & = & p(x_i|x_{f(i)})p(x_{V\setminus i}).
\end{eqnarray}

$X_{V\setminus i}$ is defined on $T'=(V\setminus i, E\setminus (f(i),i))$, a tree with $n$ nodes.
By induction assumption
\begin{eqnarray*}
p(x_{V\setminus i}) & = & p(x_s)\prod_{k\in V\setminus i\setminus s}p(x_k|x_{f(k)}).
\end{eqnarray*}
Substituting it into (\ref{eq:induction}) we have shown that Equation (\ref{eq:tree}) holds for $T$
with $|V|=n+1$. By induction argument, Equation (\ref{eq:tree}) is true for any tree.

Now for any $k\in V$, from (\ref{eq:tree}) we have
\begin{eqnarray*}
p(x_V)              & = & \Big(p(x_k|x_{f(k)})\prod_{j\in c(k)}p(x_j|x_k) \Big)\cdot q(x_{V\setminus k}), \\
p(x_{V\setminus k}) & = & \sum_{x_k} \Big(p(x_k|x_{f(k)})\prod_{j\in c(k)}p(x_j|x_k)\Big)\cdot
q(x_{V\setminus k}),
\end{eqnarray*}
where $q(x_{V\setminus k})$ is a function that does not depend on $x_k$. Then
\begin{equation*}
\begin{split}
p(x_k|x_{V\setminus k})  = & \frac{p(x_V)}{p(x_{V\setminus k})} \\
                         = & \frac{ p(x_k|x_{f(k)}) \prod_{j\in c(k)}p(x_j|x_k) }{\sum_{x_k}
                        \Big(p(x_k|x_{f(k)})\prod_{j\in c(k)}p(x_j|x_k)\Big)}\\
                         = &  \frac{p(x_{f(k)})p(x_k|x_{f(k)})\prod_{j\in
c(k)}p(x_j|x_k)}{\sum_{x_k} \Big(p(x_{f(k)})p(x_k|x_{f(k)})\prod_{j\in c(k)}p(x_j|x_k)\Big)}\\
                         = & p(x_k|x_{f(k)\cup c(k)})
\end{split}
\end{equation*}

Therefore $X_V$ is a Markov random field on $T(s,D)$. \qed
\end{proof}

For an MRF $X_V=(X_k: k\in V)$ on $T(s,D)$, we can construct an additive metric as follows. For
each link $(i,j)\in E$, we define an $M\times M$ (assume $|\mathcal{X}|=M$) \emph{forward link
transition matrix} $P_{ij}$ and an $M\times M$ \emph{backward link transition matrix} $P_{ji}$ with
entries
\begin{eqnarray*}
P_{ij}(x_i,x_j) & = & \mathbb{P}(X_j=x_j|X_{i}=x_i), \\
P_{ji}(x_j,x_i) & = & \mathbb{P}(X_i=x_i|X_{j}=x_j), \\
& &  x_i,x_j\in \mathcal{X}. \nonumber
\end{eqnarray*}

If the link transition matrices are \emph{invertible} so that
$|P_{ij}|\stackrel{\Delta}{=}|\det(P_{ij})|>0$, not equal to a \emph{permutation matrix} (a matrix
with exactly one entry in each row and each column being 1 and others being 0) so that
$|P_{ij}|<1$, and there exists a node $i\in V$ with positive marginal distribution, then we can
construct an additive metric $d_0$ with link length (\eg, \cite{Barry}, \cite{JoeChang}):
\begin{eqnarray} \label{eq:linkpara0}
d_0(e) = -\log |P_{ij}| -\log |P_{ji}|,  \phantom{aa}\forall e=(i,j)\in E.
\end{eqnarray}

For any pair of terminal nodes $i,j\in U$, the distance between $i$ and $j$ under additive metric
$d_0$ can be computed by
\begin{eqnarray}\label{eq:dmetric0}
d_0(i,j) = -\log |P_{ij}| -\log |P_{ji}|, \phantom{aa} i,j\in U.
\end{eqnarray}

We can construct other additive metrics based on the specific network inference problem. We use
link loss inference as the example. Additive metrics based on link utilization inference and link
delay inference can be found in \cite{NiTopologyInference}.
\\
\emph{Example 1 (Link Loss Inference):} For this example, the link state variable $Z_{e}$ is a
Bernoulli random variable which takes value 1 with probability $\alpha_e$ if link $e$ is in
\emph{good state} and the probe can go through the link, and takes value 0 with probability
$1-\alpha_e\stackrel{\Delta}{=}\bar{\alpha}_e$ if the probe is lost on the link (\eg,
\cite{MulticastLoss}). $\alpha_e$ is called the \emph{success rate} or \emph{packet delivery rate}
of link $e$, and $\bar{\alpha}_e$ is called the \emph{loss rate} of link $e$.

The outcome variable $X_{k}$ is also a Bernoulli random variable, which takes value 1 if the probe
successfully reaches node $k$. Since the probe is sent by the source node $s$, we have $X_s \equiv
1$. It is clear that for link loss inference
\begin{equation}
X_k = X_{f(k)}\cdot Z_{e_k}=\prod_{e\in \mathcal{P}(s,k)}Z_e.
\end{equation}

If $0<\alpha_e<1$ for all links, then we can construct an additive metric $d_l$ with link length
\begin{eqnarray} \label{eq:linkpara1}
d_l(e) & = & - \log \alpha_e, \phantom{aa} \forall e\in E.
\end{eqnarray}

Notice that there is a one-to-one mapping between the link length and link success rate, hence we
can derive the link success rates from the link lengths, and vice versa.

Under the spatial independence assumption that the link states are independent from link to link,
we have
\begin{equation*}
\begin{split}
\mathbb{P}(X_i=1)    = & \mathbb{P}(\prod_{e\in \mathcal{P}(s,i)}Z_e=1) = \prod_{e\in\mathcal{P}(s,i)}\alpha_e, \\
\mathbb{P}(X_j=1)    = & \mathbb{P}(\prod_{e\in \mathcal{P}(s,j)}Z_e=1) = \prod_{e\in\mathcal{P}(s,j)}\alpha_e, \\
\mathbb{P}(X_iX_j=1) = &
\mathbb{P}(\prod_{e\in\mathcal{P}(s,\underline{ij})}Z_e\prod_{e\in\mathcal{P}(\underline{ij},i)}Z_e\prod_{e\in\mathcal{P}(\underline{ij},j)}Z_e=1)
\\
= & \prod_{e\in\mathcal{P}(s,\underline{ij})}\alpha_e
\prod_{e\in\mathcal{P}(\underline{ij},i)}\alpha_e
\prod_{e\in\mathcal{P}(\underline{ij},j)}\alpha_e,
\end{split}
\end{equation*}
where $\underline{ij}$ is the \emph{nearest common ancestor} of $i$ and $j$ on $T(s,D)$ (\ie, the
ancestor of $i$ and $j$ that is closest to $i$ and $j$ on the routing tree). For example, in Fig.
\ref{fig:physical}(b), the nearest common ancestor of destination nodes 4 and 5 is node 2, and the
nearest common ancestor of destination nodes 4 and 6 is node 1.

Therefore, the distances between the terminal nodes, $d_l(U^2)$, can be computed by
\begin{eqnarray}\label{eq:dmetric1}
d_l(i,j) = \log \frac{\mathbb{P}(X_i=1)\mathbb{P}(X_j=1)}{\mathbb{P}^2(X_iX_j=1)},
\phantom{aa}i,j\in U.
\end{eqnarray}

\subsection{Estimation of Distances}

As in (\ref{eq:dmetric0}) and (\ref{eq:dmetric1}), if we know the pairwise joint distributions of
the outcome variables at the terminal nodes, then we can construct an additive metric and derive
$d(U^2)$. In actual network inference problems, however, the joint distributions of the outcome
variables are not given. We need to estimate the joint distributions based on measurements taken at
the terminal nodes. Specifically, the source node will send a sequence of $n$ probes, and there are
totally $n$ outcomes $X_V^{(t)}=(X^{(t)}_{k}: k\in V)$, $t=1,2,...,n$, one for each probe. For the
$t$-th probe, only the outcome variables $X^{(t)}_{U}=(X_{k}^{(t)}: k\in U=s\cup D)$ at the
terminal nodes can be measured and observed. We can estimate the joint distributions of the outcome
variables using the observed empirical distributions, which will converge to the actual
distributions almost surely if the link state processes are stationary and ergodic during the
measurement period.

Suppose $s$ sends a sequence of $n$ probes to (a subset of) destination nodes in $D$. For any
probed node $i$, let $X_i^{(t)}$ be the measured loss outcome of the $t$-th probe at node $i$, with
$X_i^{(t)}=1$ if node $i$ successfully receives the probe and $X_i^{(t)}=0$ otherwise.

We use the empirical distributions of the outcome variables to estimate the distances. For a
Bernoulli random variable $X$ (as in link loss inference), the empirical probability that $X$ takes
value 1 is just the sample mean $\bar{X}$\footnote{$\bar{X}$ is the maximum likelihood estimator
(MLE) of $\mathbb{P}(X=1)$ for the samples.} of the samples $X^{(1)},..., X^{(n)}$:
\begin{eqnarray}
\hat{P}(X=1) = \bar{X}\stackrel{\Delta}{=} \frac{1}{n}\sum_{t=1}^{n} X^{(t)}.
\end{eqnarray}

We can construct explicit estimators for the distances in (\ref{eq:dmetric1}) as follows (we use
$\hat{}$ over $d$ to represent estimated distances):
\begin{eqnarray}
\hat{d}_l(i,j) & = & \log \frac{\bar{X}_i\bar{X}_j}{\overline{X_iX_j}^2}, \label{eq:Emetric1}
\end{eqnarray}
where
\begin{eqnarray*}
\bar{X}_i & = & \frac{1}{n} \sum_{t=1}^{n}X_i^{(t)}, \\
\bar{X}_j & = & \frac{1}{n} \sum_{t=1}^{n}X_j^{(t)}, \\
\overline{X_iX_j} & = & \frac{1}{n} \sum_{t=1}^{n}X_i^{(t)}X_j^{(t)}.
\end{eqnarray*}

We can derive exponential error bounds for the distance estimators in (\ref{eq:Emetric1}) using
Chernoff bounds \cite{NTlinkestimation}.

\begin{proposition} \label{prop:disterror}
For any pair of nodes $i, j \in U$, a sample size of $n$ (number of probes to estimate
$\hat{d}_l$), and any small $\epsilon>0$:
\begin{eqnarray}
\mathbb{P}\Big\{|\hat{d}_l(i,j)-d_l(i,j)|\geq \epsilon \Big\} & \leq & e^{-c_{ij}(\epsilon)n}
\end{eqnarray}
where $c_{ij}(\epsilon)$'s are some constants.
\end{proposition}

\subsection{Other Additive Metrics and Information Fusion}

We can also construct additive metrics and compute/estimate the distances between the terminal
nodes using (end-to-end) unicast packet pair probing or traceroute probing, as described in
\cite{NiTopologyInference}. A nice property of additive metrics is that a linear combination of
several additive metrics is still an additive metric. In order to fuse information collected from
different measurements, we can construct a new additive metric using a linear (convex) combination
of additive metrics $d_1, d_2,..., d_k$:
\begin{eqnarray}
d  =  a_1d_1+a_2d_2+...+a_kd_k, \\
\mbox{ s.t. }  a_1+a_2+...+a_k=1. \nonumber
\end{eqnarray}

The estimated distance between terminal nodes $i,j \in U$ under the new additive metric can be
easily computed:
\begin{eqnarray*}
\hat{d}(i,j) = a_1\hat{d}_1(i,j)+a_2 \hat{d}_2(i,j)+...+a_k\hat{d}_k(i,j).
\end{eqnarray*}

In practice we can select the coefficients empirically based on the current network state or to
minimize the variance of the estimator $\hat{d}$.

\section{Neighbor-Joining Algorithm}  \label{sec:NTNJ}

We have described how to construct additive metrics and estimate the distances between the terminal
nodes via end-to-end packet probing measurements. In this section we introduce the
\emph{neighbor-joining} (NJ) algorithm, which is considered the most widely used algorithm for
building binary phylogenetic trees from distances (\eg, \cite{NJRevealed}, \cite{NeighborJoining},
\cite{ProspectsNJ}).

\begin{definition}
A \textbf{\emph{distance-based tree inference algorithm}} (or distance-based algorithm for short)
takes the (estimated) distances between the terminal nodes of a tree as the input and returns a
tree topology and the associated link lengths. The input distances $\hat{d}(U^2)$ satisfy:
\begin{eqnarray*}
\hat{d}(i,j) & \geq & 0, \phantom{aa} \mbox{\emph{with equality if and only if} } i=j,\\
\hat{d}(i,j) & = & \hat{d}(j,i).
\end{eqnarray*}
\end{definition}

\begin{definition}
Two or more nodes on a tree are called \emph{\textbf{neighbors}} (\emph{\textbf{siblings}}), if
they are connected via one internal node (if they have the same parent) on the tree.
\end{definition}

The NJ algorithm is an \emph{agglomerative algorithm}. The algorithm begins with a leaf set
including all destination nodes. In each step it selects two leaf nodes that are likely to be
neighbors, deletes them from the leaf set, creates a new node as their parent and adds that node to
the leaf set. The whole process is iterated until there is only one node left in the leaf set,
which will be the child of the root (source node).

To avoid trivial cases, we assume $|D|\geq 2$.
\\\\
\setlength{\unitlength}{1mm}
\begin{picture}(89,0)
\put(0,0){\line(1,0){89}}  \thicklines
\end{picture}
\\
 \emph{Algorithm 1: Neighbor-Joining (NJ) Algorithm for Binary Trees}
\\
\setlength{\unitlength}{1mm}
\begin{picture}(89,0)
\put(0,0){\line(1,0){89}}  \thicklines
\end{picture}\\
\textbf{Input:} Estimated distances between the nodes in $U=s\cup D$, $\hat{d}(U^2)$.
\begin{itemize}
\item[1.] $V=\emptyset$, $E=\emptyset$.

\item[2.1] For any pair of nodes $i,j\in D$, compute
          \begin{eqnarray} \label{eq:NJQfun}
          \hat{Q}(i,j) = \sum_{k\in U}\hat{d}(i,k)+\sum_{k\in U}\hat{d}(j,k)-(|U|-2)\hat{d}(i,j).
          \end{eqnarray}

\item[2.2] Find $i^*,j^*\in D$ with the largest $\hat{Q}(i,j)$ (break the tie arbitrarily).\\
Create a node $f$ as the parent of $i^*$ and $j^*$.\\
$D = D\setminus \{i^*,j^*\}$, $U = U\setminus \{i^*,j^*\}$\\
$V = V\cup \{i^*,j^*\}$, $E = E\cup \{(f, i^*),(f, j^*)\}$.

\item[2.3] Compute the link lengths from the distances:
\begin{eqnarray}
\hat{d}(f,i^*) = \frac{1}{|U|}\sum_{k\in U}
\big[\hat{d}(k,i^*)+\hat{d}(i^*,j^*)-\hat{d}(k,j^*)\big]/2, \\
\hat{d}(f,j^*) = \frac{1}{|U|}\sum_{k\in U}
\big[\hat{d}(k,j^*)+\hat{d}(i^*,j^*)-\hat{d}(k,i^*)\big]/2.
\end{eqnarray}

\item[2.4] For each $k\in U$, compute the distance between $k$ and $f$:
\begin{eqnarray}
\hat{d}(k,f) = \frac{1}{2}\big[\hat{d}(k,i^*)-\hat{d}(f,i^*)\big] +
\frac{1}{2}\big[\hat{d}(k,j^*)-\hat{d}(f,j^*)\big].
\end{eqnarray}
$D = D \cup f$, $U = U \cup f$.

\item[3.] If $|D|=1$, for the $i \in D$: $V=V\cup \{i\}$, $E=E\cup(s,i)$.
\\Otherwise, repeat Step 2.

\end{itemize}
\textbf{Output:} Tree $\hat{T}=(V,E)$, and link lengths $\hat{d}(e)$ for all $e\in E$.\\
\setlength{\unitlength}{1mm}
\begin{picture}(89,0)
\put(0,0){\line(1,0){89}}  \thicklines
\end{picture}
\\

The NJ algorithm has several nice properties:
\begin{itemize}
\item it is computationally efficient, with a polynomial-time complexity $O(N^3)$ for (binary)
trees with $N$ terminal nodes;

\item it returns the correct tree topology and link lengths if the input distances are
\emph{additive} (\ie, if the input distances are derived from an additive metric without estimation
errors);

\item it is robust: it achieves the optimal $l_\infty$-\emph{radius} among all distance-based
algorithms for binary trees.
\end{itemize}

The $l_\infty$-radius notation was introduced by Atteson \cite{Atteson}.

\begin{definition}
For a distance-based algorithm, we say it has $l_\infty$-\textbf{\emph{radius}} $r$, if for any
tree $T$ associated with any additive metric $d$, whenever the input distances between the terminal
nodes, $\hat{d}(U^2)$, satisfy:
\begin{eqnarray}
||\hat{d}(U^2)-d(U^2)||_\infty & \stackrel{\Delta}{=} & \max_{i,j\in U}|\hat{d}(i,j)-d(i,j)|
\nonumber \\
& < & r\min_{e\in E}d(e),
\end{eqnarray}
the algorithm will return the correct topology of $T$.
\end{definition}

An algorithm with larger $l_\infty$-radius is more robust, because it can tolerate more estimation
errors. \cite{Atteson} showed that no distance-based algorithm has $l_\infty$-radius larger than
$\frac{1}{2}$ via an example, and proved that the NJ algorithm in fact achieves the optimal
$l_\infty$-radius for binary trees.

\begin{theorem}
The NJ algorithm achieves the optimal $l_\infty$-radius $\frac{1}{2}$ for binary trees.
\end{theorem}

It is not straightforward to extend the NJ algorithm for general (non-binary) trees. Since most
routing trees in communication networks are not binary, we are motivated to design algorithms that
can handle general trees.

\section{Rooted Neighbor-Joining Algorithm} \label{sec:NTRNJ}

\subsection{Binary Trees}

We first present an algorithm which can be viewed as a \emph{rooted} version of the NJ algorithm
for binary trees. To avoid trivial cases, we assume $|D|\geq 2$.
\\\\
\setlength{\unitlength}{1mm}
\begin{picture}(89,0)
\put(0,0){\line(1,0){89}}  \thicklines
\end{picture}
\\
 \emph{Algorithm 2: Rooted Neighbor-Joining (RNJ) Algorithm for Binary Trees}
\\
\setlength{\unitlength}{1mm}
\begin{picture}(89,0)
\put(0,0){\line(1,0){89}}  \thicklines
\end{picture}
\\
\textbf{Input:} Estimated distances between the nodes in $U=s\cup D$, $\hat{d}(U^2)$.
\begin{itemize}
\item[1.] $V=\{s\}$, $E=\emptyset$.\\
          For any pair of nodes $i,j\in D$, compute
          \begin{eqnarray}\label{eq:RNJrhofun}
          \hat{\rho}(i,j) & = & \frac{\hat{d}(s,i)+\hat{d}(s,j)-\hat{d}(i,j)}{2}.
          \end{eqnarray}

\item[2.1] Find $i^*,j^*\in D$ with the largest $\hat{\rho}(i,j)$ (break the tie arbitrarily).
\\Create a node $f$ as the parent of $i^*$ and $j^*$.\\
$D=D\setminus \{i^*,j^*\}$,\\
$V=V\cup \{i^*,j^*\}$, $E=E\cup \{(f, i^*),(f, j^*)\}$.

\item[2.2] Compute:
\begin{eqnarray}
\hat{d}(s,f) & = & \hat{\rho}(i^*,j^*), \\
\hat{d}(f,i^*) & = & \hat{d}(s,i^*)-\hat{\rho}(i^*,j^*), \\
\hat{d}(f,j^*) & = & \hat{d}(s,j^*)-\hat{\rho}(i^*,j^*).
\end{eqnarray}

\item[2.3] For each $k\in D$, compute:
\begin{equation}
\begin{split}
\hat{d}(k,f)  = & \frac{1}{2}\big[\hat{d}(k,i^*)-\hat{d}(f,i^*)\big] +
\frac{1}{2}\big[\hat{d}(k,j^*)-\hat{d}(f,j^*)\big], \\
\hat{\rho}(k,f) = & \frac{1}{2}\big[\hat{d}(s,k)+\hat{d}(s,f)-\hat{d}(k,f)\big] \nonumber \\
                = & \frac{1}{2}\big[\hat{\rho}(k,i^*)+\hat{\rho}(k,j^*)\big].
\end{split}
\end{equation}
$D = D\cup f$.

\item[3.] If $|D|=1$, for the $i \in D$: $V=V\cup \{i\}$, $E=E\cup(s,i)$.
\\Otherwise, repeat Step 2.

\end{itemize}
\textbf{Output:} Tree $\hat{T}=(V,E)$, and link lengths $\hat{d}(e)$ for all $e\in E$.\\

\begin{picture}(89,0)
\put(0,0){\line(1,0){89}}  \thicklines
\end{picture}
\\

The major difference between the NJ algorithm and the RNJ algorithm is the selection of the
\emph{score function}: the NJ algorithm uses the $\hat{Q}$ function defined in (\ref{eq:NJQfun}),
which has no simple interpretation; while the RNJ algorithm uses the $\hat{\rho}$ function in
(\ref{eq:RNJrhofun}), which has a simple interpretation that we will explain next.

For any pair of nodes $i, j\in D$, remember $\underline{ij}$ is their nearest common ancestor on
$T(s,D)$. Under additive metric $d$, we know
\begin{equation} \label{eq:rho}
\rho(i,j) = \frac{d(s,i)+d(s,j)-d(i,j)}{2} = d(s,\underline{ij})
\end{equation}
is the distance from the root (source node $s$) to $\underline{ij}$. It is not hard to verify that
a pair of nodes $i^*, j^*$ with largest $\rho(i,j)$ must be neighbors (siblings) on the tree.
$\hat{\rho}(i,j)$ in (\ref{eq:RNJrhofun}) is the estimated distance from the root to
$\underline{ij}$ computed from the input distances. If the input distances are close to the true
additive distances, then we would expect that the two nodes selected in Step 2.1 of Algorithm 2 are
indeed neighbors.

We provide a sufficient condition for Algorithm 2 to return the correct tree topology. From the
condition we can establish several nice properties of the algorithm.

\begin{lemma} \label{lemma:RNJBSuff}
For binary trees, a sufficient condition for Algorithm 2 to return the correct tree topology is:
\begin{eqnarray}\label{eq:topocon1}
&  & \forall i,j,k \in D \mbox{ \emph{s.t.} } \underline{ij} \prec \underline{ik} \nonumber \\
\Rightarrow &  & \hat{\rho}(i,j) > \hat{\rho}(i,k),
\end{eqnarray}
where $\underline{ij} \prec \underline{ik}$ means that $\underline{ij}$ is descended from
$\underline{ik}$.
\end{lemma}

\begin{proof}
We prove the lemma by induction on the cardinality of $D$.
\\ (1) If $|D|=2$, then clearly Algorithm 2 will return the
correct tree topology.
\\ (2) Assume Algorithm 2 returns the correct tree topology under condition (\ref{eq:topocon1}) for
$|D|\leq N$. Now consider $|D|=N+1$.

\emph{\textbf{Claim 1. $i^*,j^*$ found in Step 2.1 which maximize $\hat{\rho}(i,j)$  are siblings
(neighbors).}}\\
If $i^*$ and $j^*$ are not siblings, then there exists $k \in D$ such that either
$\underline{i^*k}$ or $\underline{j^*k}$ is descended from $\underline{i^*j^*}$. Under condition
(\ref{eq:topocon1}), this implies either
\begin{eqnarray*}
& & \hat{\rho}(i^*, k) > \hat{\rho}(i^*, j^*) \\
\mbox{ or } & & \hat{\rho}(j^*, k) >  \hat{\rho}(i^*, j^*),
\end{eqnarray*}
a contradiction to the maximality of $\hat{\rho}(i^*, j^*)$.

\emph{\textbf{Claim 2. Condition (\ref{eq:topocon1}) is maintained over the nodes in $D$ after Step
2.}}
\\
After Step 2, $i^*,j^*$ are deleted from $D$ and $f$ is added to $D$ as a new leaf node. Since
$i^*,j^*$ are siblings and $f$ is their parent, we know that for any $i \in D$,
\begin{equation*}
\underline{if}=\underline{ii^*}=\underline{ij^*}.
\end{equation*}

Therefore, $\forall i,j \in D$ s.t. $\underline{ij} \prec \underline{if}$, we have $\underline{ij}
\prec \underline{ii^*}$ and $\underline{ij} \prec \underline{ij^*}$, which implies
\begin{eqnarray*}
& & \hat{\rho}(i,j) > \hat{\rho}(i,i^*)\\
\mbox{ and }   & & \hat{\rho}(i,j) > \hat{\rho}(i,j^*)\\
\mbox{ hence } & & \hat{\rho}(i,j) > \hat{\rho}(i,f) =
\frac{1}{2}\big[\hat{\rho}(i,i^*)+\hat{\rho}(i,j^*)\big].
\end{eqnarray*}

Similarly, $\forall i,k \in D$ s.t. $\underline{if} \prec \underline{ik}$, we can show
$\hat{\rho}(i,f) > \hat{\rho}(i,k)$.

From claims 1 and 2, we know that after one iteration of Step 2, Algorithm 2 will correctly find
out a pair of siblings, and condition (\ref{eq:topocon1}) is maintained for the new set of leaf
nodes in $D$. Then $|D|$ is decreased by 1. By induction assumption, the algorithm will return the
correct topology of the remaining part of the tree. This completes our proof of the lemma. \qed
\end{proof}

\begin{proposition} \label{prop:RNJBExact}
For binary trees, Algorithm 2 will return the correct tree topology and link lengths if the input
distances $\hat{d}(U^2)$ are additive.
\end{proposition}
\begin{proof}
If the input distances are additive, then $\hat{\rho}(i,j)$ and $\hat{\rho}(i,k)$ are the actual
distances from $s$ to $\underline{ij}$ and $\underline{ik}$ under an additive metric. In this case,
if $\underline{ij}$ is descended from $\underline{ik}$, since link lengths are positive, we have
$\hat{\rho}(i,j) > \hat{\rho}(i,k)$, hence condition (\ref{eq:topocon1}) holds. Then by Lemma
\ref{lemma:RNJBSuff}, Algorithm 2 will return the correct tree topology. In addition, under
additive distances it is clear that the link lengths computed in Step 2.2 of Algorithm 2 are
correct. \qed
\end{proof}

In practice, the distances between the terminal nodes are estimated from measurements taken at the
terminal nodes. The estimated distances may deviate from the true additive distances due to
measurement errors. Nevertheless, we will show that if the estimated distances are close enough to
the true distances, then Algorithm 2 will return the correct tree topology. In addition, Algorithm
2 achieves the optimal $l_\infty$-radius among all distance-based algorithms.

\begin{proposition} \label{prop:RNJBApprox}
The RNJ algorithm (Algorithm 2) achieves the optimal $l_\infty$-radius $\frac{1}{2}$ for binary
trees, \ie, for any binary tree associated with any additive metric $d$, whenever the input
distances $\hat{d}(U^2)$ satisfy:
\begin{eqnarray} \label{eq:a2radius}
\max_{i,j\in U}|\hat{d}(i,j)-d(i,j)| & < & \frac{1}{2}\min_{e\in E}d(e),
\end{eqnarray}
Algorithm 2 will return the correct tree topology.
\end{proposition}

\begin{proof}
Using Lemma \ref{lemma:RNJBSuff} we only need to show that condition (\ref{eq:a2radius}) implies
condition (\ref{eq:topocon1}). Let
\begin{eqnarray*}
\Delta & = & \min_{e\in E}d(e)
\end{eqnarray*}
be the minimum link length on the tree. If $\underline{ij} \prec \underline{ik}$, \ie, if
$\underline{ij}$ is descended from $\underline{ik}$, since link lengths $\geq$ $\Delta$, we have
\begin{eqnarray*}
\rho(i,j)-\rho(i,k) & \geq & \Delta.
\end{eqnarray*}
Then from (\ref{eq:RNJrhofun}), (\ref{eq:rho}), (\ref{eq:a2radius}) we have:
\begin{eqnarray*}
 &      & \hat{\rho}(i,j)-\hat{\rho}(i,k) \\
 & \geq & \Big(\hat{\rho}(i,j)-\hat{\rho}(i,k)\Big) - \Big(\rho(i,j)-\rho(i,k)-\Delta \Big) \\
 & \geq & \Delta - \frac{1}{2}|(\hat{d}(s,j)-d(s,j)| - \frac{1}{2}|\hat{d}(i,j)-d(i,j)| \\
 &      & -\frac{1}{2}|\hat{d}(s,k)-d(s,k)| - \frac{1}{2}|\hat{d}(i,k)-d(i,k)| \\
 &  >   & \Delta - \frac{1}{4}\Delta - \frac{1}{4}\Delta - \frac{1}{4}\Delta - \frac{1}{4}\Delta \\
 &  >   &  0.
\end{eqnarray*}

Hence condition (\ref{eq:a2radius}) indeed implies condition (\ref{eq:topocon1}). Since
(\ref{eq:topocon1}) is a sufficient condition for Algorithm 2 to return the correct tree topology,
(\ref{eq:a2radius}) is also a sufficient condition for Algorithm 2 to return the correct tree
topology. \qed
\end{proof}

\subsection{General Trees}

$\rho(i,j)$ is the distance from the root to the nearest common ancestor of nodes $i$ and $j$. For
a general routing tree with positive link lengths, we have several observations of the $\rho$
function.
\begin{itemize}
\item If nodes $i$ and $j$ are neighbors on the tree, then for any other node $k$ on the tree we
have
\begin{eqnarray}
\rho(i,j) & \geq & \rho(i,k).
\end{eqnarray}

\item If nodes $i$ and $j$ are neighbors on the tree, then for any other node $k$ that is also a
neighbor of $i$ and $j$ we have
\begin{eqnarray} \label{eq:obsneighbor}
\rho(i,j) & = & \rho(i,k)
\end{eqnarray}
because $\underline{ij}=\underline{ik}$.

\item If nodes $i$ and $j$ are neighbors on the tree, then for any other node $k$ that is not a
neighbor of $i$ and $j$ we have
\begin{eqnarray} \label{eq:obsnonneighbor}
\rho(i,j) & \geq & \rho(i,k)+\Delta
\end{eqnarray}
(where $\Delta$ is the minimum link length) because $\underline{ij}$ is descended from
$\underline{ik}$ and they are separated by at least one link.
\end{itemize}

Therefore, we can determine whether a group of nodes are neighbors on the tree from knowledge of
the $\rho$ function under an additive metric.

To extend the RNJ algorithm (Algorithm 2) for general trees, after we find out two nodes $i^*$ and
$j^*$ that are likely to be neighbors in Step 2.1, we need to find out other nodes that are likely
to be neighbors of $i^*$ and $j^*$ based on $\hat{\rho}$ computed from the input distances. We use
the following \emph{threshold neighbor criterion}:
\\\\
\textbf{Threshold Neighbor Criterion.}\\ Suppose $i^*$ and $j^*$ are neighbors on the tree. Node
$k$ will be chosen as a neighbor of $i^*$ and $j^*$ if and only if
\begin{eqnarray} \label{eq:thresholdcriterion}
\hat{\rho}(i^*,j^*) - \hat{\rho}(i^*,k) & \leq & t
\end{eqnarray}
for some threshold $t>0$.\\

Based on observations (\ref{eq:obsneighbor}) and (\ref{eq:obsnonneighbor}), and since $\hat{\rho}$
is an estimator of $\rho$ with possible estimation errors, we use the middle point
$\frac{\Delta}{2}$ as the threshold. Later we will show that such a threshold enables the algorithm
to achieve the optimal $l_\infty$-radius $\frac{1}{4}$ for general trees if the threshold criterion
(\ref{eq:thresholdcriterion}) is used in the algorithm (see the proof of Proposition
\ref{prop:upperboundradius}).
\\\\
\setlength{\unitlength}{1mm}
\begin{picture}(89,0)
\put(0,0){\line(1,0){89}}  \thicklines
\end{picture}
\\
 \emph{Algorithm 3: Rooted Neighbor-Joining (RNJ) Algorithm for General Trees}
\\
\setlength{\unitlength}{1mm}
\begin{picture}(89,0)
\put(0,0){\line(1,0){89}}  \thicklines
\end{picture}\\
\textbf{Input:} Estimated distances between the nodes in $U=s\cup D$, $\hat{d}(U^2)$; estimated
minimum link length $\Delta>0$.
\begin{itemize}
\item[1.] $V=\{s\}$, $E=\emptyset$.\\
          For any pair of nodes $i,j\in D$, compute
          \begin{eqnarray}\label{eq:RNJGrhofun}
          \hat{\rho}(i,j) & = & \frac{\hat{d}(s,i)+\hat{d}(s,j)-\hat{d}(i,j)}{2}.
          \end{eqnarray}

\item[2.1] Find $i^*,j^*\in D$ with the largest $\hat{\rho}(i,j)$ (break the tie arbitrarily).
\\Create a node $f$ as the parent of $i^*$ and $j^*$.\\
$D=D\setminus \{i^*,j^*\}$, \\
$V=V\cup \{i^*,j^*\}$, $E=E\cup \{(f, i^*),(f, j^*)\}$.

\item[2.2] Compute:
\begin{eqnarray}
\hat{d}(s,f)   & = & \hat{\rho}(i^*,j^*), \\
\hat{d}(f,i^*) & = & \hat{d}(s,i^*)-\hat{\rho}(i^*,j^*), \\
\hat{d}(f,j^*) & = & \hat{d}(s,j^*)-\hat{\rho}(i^*,j^*).
\end{eqnarray}

\item[2.3] For every $k\in D$ such that $\hat{\rho}(i^*,j^*)-\hat{\rho}(i^*,k) \leq
\frac{\Delta}{2}$:
\\
$D=D\setminus k$, \\
$V=V\cup k$, $E=E\cup (f, k)$.\\
Compute:
\begin{eqnarray}
\hat{d}(f,k) & = & \hat{d}(s,k)- \hat{\rho}(i^*,j^*).
\end{eqnarray}

\item[2.4] For each $k\in D$, compute:
\begin{equation}
\begin{split}
\hat{d}(k,f) = & \frac{1}{2}\big[\hat{d}(k,i^*)-\hat{d}(f,i^*)\big] +
\frac{1}{2}\big[\hat{d}(k,j^*)-\hat{d}(f,j^*)\big], \\
\hat{\rho}(k,f) = & \frac{1}{2}\big[\hat{d}(s,k)+\hat{d}(s,f)-\hat{d}(k,f)\big] \nonumber \\
                = & \frac{1}{2}\big[\hat{\rho}(k,i^*)+\hat{\rho}(k,j^*)\big].
\end{split}
\end{equation}
$D = D\cup f$.

\item[3.] If $|D|=1$, for the $i \in D$: $V=V\cup \{i\}$, $E=E\cup(s,i)$.
\\Otherwise, repeat Step 2.

\end{itemize}
\emph{Output:} Tree $\hat{T}=(V,E)$, and link lengths $\hat{d}(e)$ for all $e\in E$.\\
\begin{picture}(89,0)
\put(0,0){\line(1,0){89}}  \thicklines
\end{picture}
\\

\begin{lemma}  \label{lemma:RNJGSuff}
Let $\Delta \leq \min_{e\in E}d(e)$ be the input parameter. A sufficient condition for Algorithm 3
to return the correct tree topology is:
\begin{eqnarray} \label{eq:topocon2}
\forall i,j,k \in D \mbox{ s.t. } \underline{ij} \prec \underline{ik} & \Rightarrow &
\hat{\rho}(i,j) - \hat{\rho}(i,k) > \frac{\Delta}{2}, \nonumber \\
\forall i,j,k \in D \mbox{ s.t. } \underline{ij} = \underline{ik} & \Rightarrow & |\hat{\rho}(i,j)
- \hat{\rho}(i,k)| \leq \frac{\Delta}{2}. \nonumber \\
\end{eqnarray}
\end{lemma}

\begin{proof}
We outline the proof, which is similar to the proof of Lemma \ref{lemma:RNJBSuff}. There are three
key observations:
\begin{itemize}
\item[(1)] Under condition (\ref{eq:topocon2}), $i^*,j^*$ found in Step 2.1 of Algorithm 3 are
siblings.

\item[(2)] Under condition (\ref{eq:topocon2}), $k$ will be selected in Step 2.3 if and only if it
is a sibling of $i^*$ and $j^*$.

\item[(3)] Condition (\ref{eq:topocon2}) is maintained over the nodes in $D$ after Step 2.
\end{itemize}

The lemma then follows by induction on the cardinality of $D$. \qed
\end{proof}

\begin{proposition}
For general trees, Algorithm 3 will return the correct tree topology and link lengths if the input
distances $\hat{d}(U^2)$ are additive.
\end{proposition}
\begin{proof}
The proof is similar to the proof of Proposition \ref{prop:RNJBExact}.
\end{proof}

In practice the input distances may deviate from the true additive distances due to measurement
errors. Again we can show that if the input distances are close enough to the true additive
distances, then Algorithm 3 will return the correct tree topology.

\begin{proposition} \label{prop:RNJGApprox}
For a general tree with additive metric $d$, if the input parameter
\begin{eqnarray*}
\Delta & \leq & \min_{e\in E}d(e)
\end{eqnarray*}
and the input distances $\hat{d}(U^2)$ satisfy:
\begin{eqnarray} \label{eq:topocon3}
\max_{i,j\in U}|\hat{d}(i,j)-d(i,j)| & < & \frac{\Delta}{4},
\end{eqnarray}
then Algorithm 3 will return the correct tree topology.
\end{proposition}

\begin{proof}
The proof is similar to the proof of Proposition \ref{prop:RNJBApprox}. We can show that condition
(\ref{eq:topocon3}) implies condition (\ref{eq:topocon2}), then the proposition follows by Lemma
\ref{lemma:RNJGSuff}. \qed
\end{proof}

If the input parameter $\Delta = \min_{e\in E}d(e)$, then Proposition \ref{prop:RNJGApprox} says
that the RNJ algorithm has $l_\infty$-radius $\frac{1}{4}$ for general trees.

\begin{corollary}
The RNJ algorithm (Algorithm 3) has $l_\infty$-radius $\frac{1}{4}$ for general trees when $\Delta
= \min_{e\in E}d(e)$.
\end{corollary}

We have the following conjecture.

\begin{conjecture}
No distance-based algorithm has $l_\infty$-radius greater than $\frac{1}{4}$ for general trees. If
this is true, then the RNJ algorithm (Algorithm 3) achieves the optimal $l_\infty$-radius
$\frac{1}{4}$ for general trees when $\Delta = \min_{e\in E}d(e)$.
\end{conjecture}

We can show that no distance-based algorithm has $l_\infty$-radius greater than $\frac{1}{4}$ if
the threshold (neighbor) criterion (\ref{eq:thresholdcriterion}) is used in the algorithm.

\begin{proposition} \label{prop:upperboundradius}
If the threshold criterion (\ref{eq:thresholdcriterion}) is used, then no distance-based algorithm
has $l_\infty$-radius greater than $\frac{1}{4}$ for general trees.
\end{proposition}
\begin{proof}
Suppose $\mathcal{A}$ is a distance-based algorithm with $l_\infty$-radius $r$ in which the
threshold criterion (\ref{eq:thresholdcriterion}) is used. Let $\Delta = \min_{e\in E}d(e)$.
Therefore, for any tree $T$ associated with any additive metric $d$, if the input distances
$\hat{d}(U^2)$ satisfy:
\begin{eqnarray} \label{eq:l_inftyradiusA}
\max_{i,j\in U}|\hat{d}(i,j)-d(i,j)| & < & r\Delta,
\end{eqnarray}
then $\mathcal{A}$ will return the correct topology of $T$.

Suppose $i^*$ and $j^*$ are neighbors on $T$, and $k$ is a neighbor of them. Then we have
$\rho(i^*,j^*)=\rho(i^*,k)$. Under condition (\ref{eq:l_inftyradiusA}) we know
\begin{eqnarray*}
\hat{\rho}(i^*,j^*)-\hat{\rho}(i^*,k) < 2r\Delta.
\end{eqnarray*}
Since the threshold criterion (\ref{eq:thresholdcriterion}) is used, we need to have
\begin{eqnarray} \label{eq:r*con1}
\hat{\rho}(i^*,j^*)-\hat{\rho}(i^*,k) < 2r\Delta \leq t & \Rightarrow & r \leq \frac{t}{2\Delta}
\end{eqnarray}
to correctly add $k$ as a neighbor of $i^*$ and $j^*$.

Now suppose $k'$ is not a neighbor of $i^*$ and $j^*$. Then we have $\rho(i^*,j^*)-\rho(i^*,k')\geq
\Delta$. Under condition (\ref{eq:l_inftyradiusA}) we know
\begin{eqnarray*}
\hat{\rho}(i^*,j^*)-\hat{\rho}(i^*,k') > \Delta-2r\Delta.
\end{eqnarray*}
Since the threshold criterion (\ref{eq:thresholdcriterion}) is used, we need to have
\begin{eqnarray} \label{eq:r*con2}
\hat{\rho}(i^*,j^*)-\hat{\rho}(i^*,k') > \Delta-2r\Delta \geq t & \Rightarrow & r \leq
\frac{1}{2}-\frac{t}{2\Delta} \nonumber \\
\end{eqnarray}
to correctly not add $k'$ as a neighbor of $i^*$ and $j^*$.

Combining (\ref{eq:r*con1}) and (\ref{eq:r*con2}) we have
\begin{eqnarray} \label{eq:r*con}
r \leq \min\Big(\frac{t}{2\Delta}, \frac{1}{2}-\frac{t}{2\Delta}\Big) & \Rightarrow & r \leq
\frac{1}{4}
\end{eqnarray}
where the upper bound $\frac{1}{4}$ of $r$ is achieved with the threshold $t=\frac{\Delta}{2}$.
\qed
\end{proof}

\subsection{Complexity and Consistency}

The computational complexity of the RNJ algorithm is $O(N^2\log N)$ for a routing tree with $N$
destination nodes. We now show the \emph{consistency} of the RNJ algorithm for general trees
(Algorithm 3), and a similar result holds for binary trees.

Let $\hat{T}_n$ be the inferred tree topology returned by the RNJ algorithm with a sample size $n$
(number of probes to estimate the distances between the terminal nodes). Let
\begin{eqnarray*}
P_n & = & \mathbb{P}\{\hat{T}_n= T\}
\end{eqnarray*}
denote the probability of correct topology inference of the RNJ algorithm.

\begin{proposition} \label{prop:RNJsamplesize}
Let $\Delta \leq \min_{e\in E}d(e)$ be the input parameter of the RNJ algorithm. If
\begin{eqnarray} \label{eq:disterror2}
\mathbb{P}\Big\{|\hat{d}(i,j)-d(i,j)|\geq \frac{\Delta}{4}\Big\} \leq  e^{-c_{ij}(\Delta)n},
\forall i,j\in U,
\end{eqnarray}
where $n$ is the sample size and $c_{ij}(\Delta)$ is a constant, then for a routing tree with $N$
terminal nodes:
\begin{eqnarray}
P_n & \geq & 1- N^2e^{-c(\Delta)n}.
\end{eqnarray}
\end{proposition}

\begin{proof}
By Proposition \ref{prop:RNJGApprox} we have
\begin{eqnarray*} P_n & \geq & \mathbb{P}\Big\{
\bigcap_{i,j\in U} |\hat{d}(i, j)-d(i, j)| < \frac{\Delta}{4} \Big\}
\nonumber\\
& = & 1- \mathbb{P}\Big\{ \bigcup_{i,j\in U} |\hat{d}(i,j)-d(i,j)|\geq \frac{\Delta}{4} \Big\} \nonumber \\
& \geq & 1-\sum_{i,j\in U} e^{-c_{ij}(\Delta)n} \\
& \geq &  1-N^2e^{-c(\Delta)n}
\end{eqnarray*}
where $C(\Delta)=\min_{i,j\in U}c_{ij}(\Delta)$. \qed
\end{proof}

\begin{proposition} \label{prop:RNJlinklength}
If the input distances $\hat{d}(U^2)$ are consistent (\ie, they converge to the true distances in
probability in the sample size) and the RNJ algorithm returns the correct tree topology, then the
link lengths returned by the RNJ algorithm are consistent.
\end{proposition}

If we use the distance estimators in (\ref{eq:Emetric1}), since they satisfy condition
(\ref{eq:disterror2}) (by Proposition \ref{prop:disterror}) and are consistent, by Proposition
\ref{prop:RNJsamplesize}, the probability of correct topology inference of the RNJ algorithm goes
to 1 exponentially fast in the sample size. If the inferred topology is correct, then by
Proposition \ref{prop:RNJlinklength}, the returned link lengths are also consistent. For network
inference problems where there is a one-to-one mapping between the link performance parameters and
the link lengths (\eg, (\ref{eq:linkpara1})), the link lengths returned by the RNJ algorithm
provide consistent estimators for the link performance parameters (\eg, success rates).

\begin{figure}
\center{\psfig{file=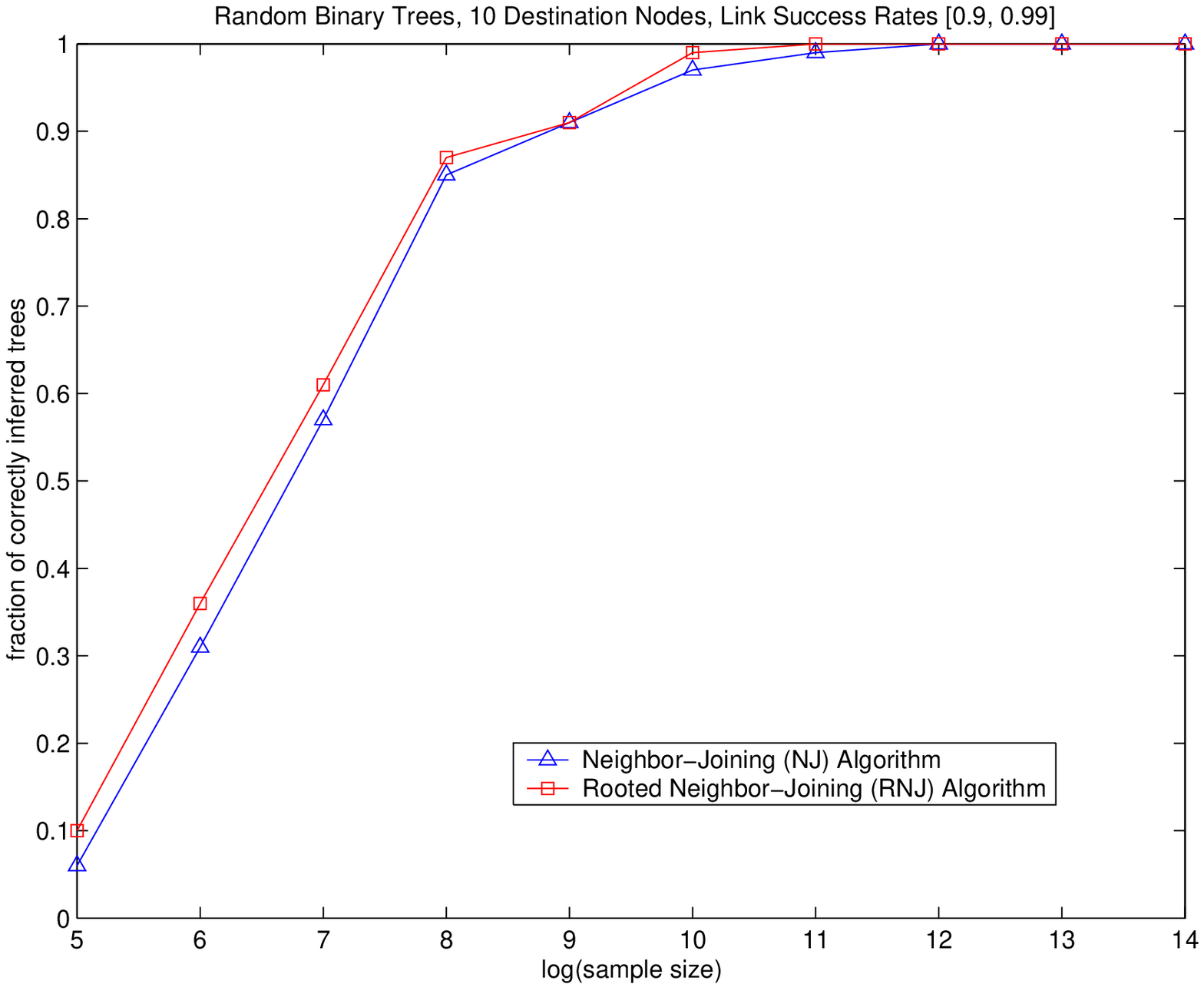, width=\columnwidth}} \caption{Binary trees: fraction of
correctly inferred trees.} \label{fig:MBN10a90topo}

\center{\psfig{file=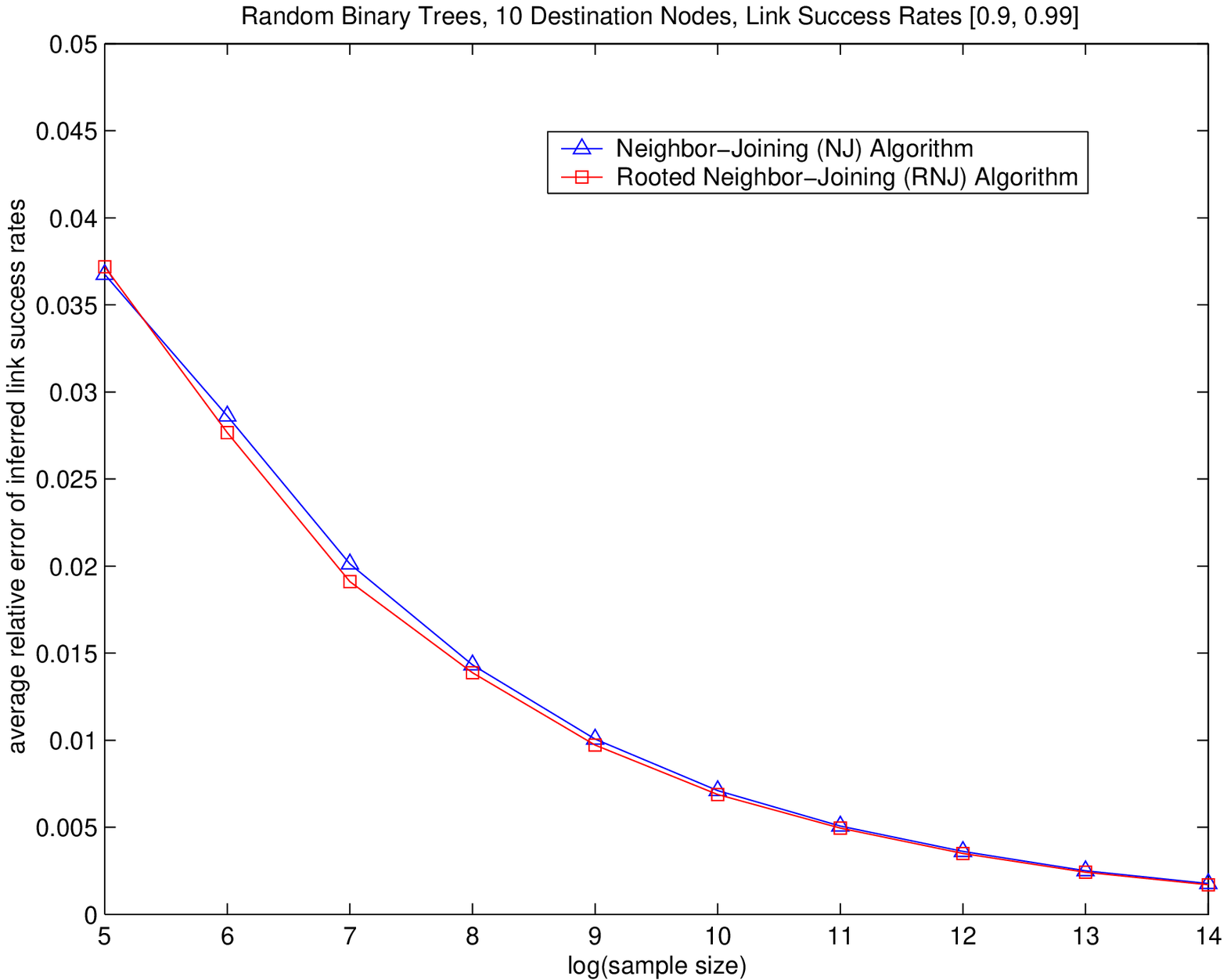, width=\columnwidth}} \caption{Binary trees: average
relative error of inferred link success rates.} \label{fig:MBN10a90link}
\end{figure}

\begin{figure}
\center{\psfig{file=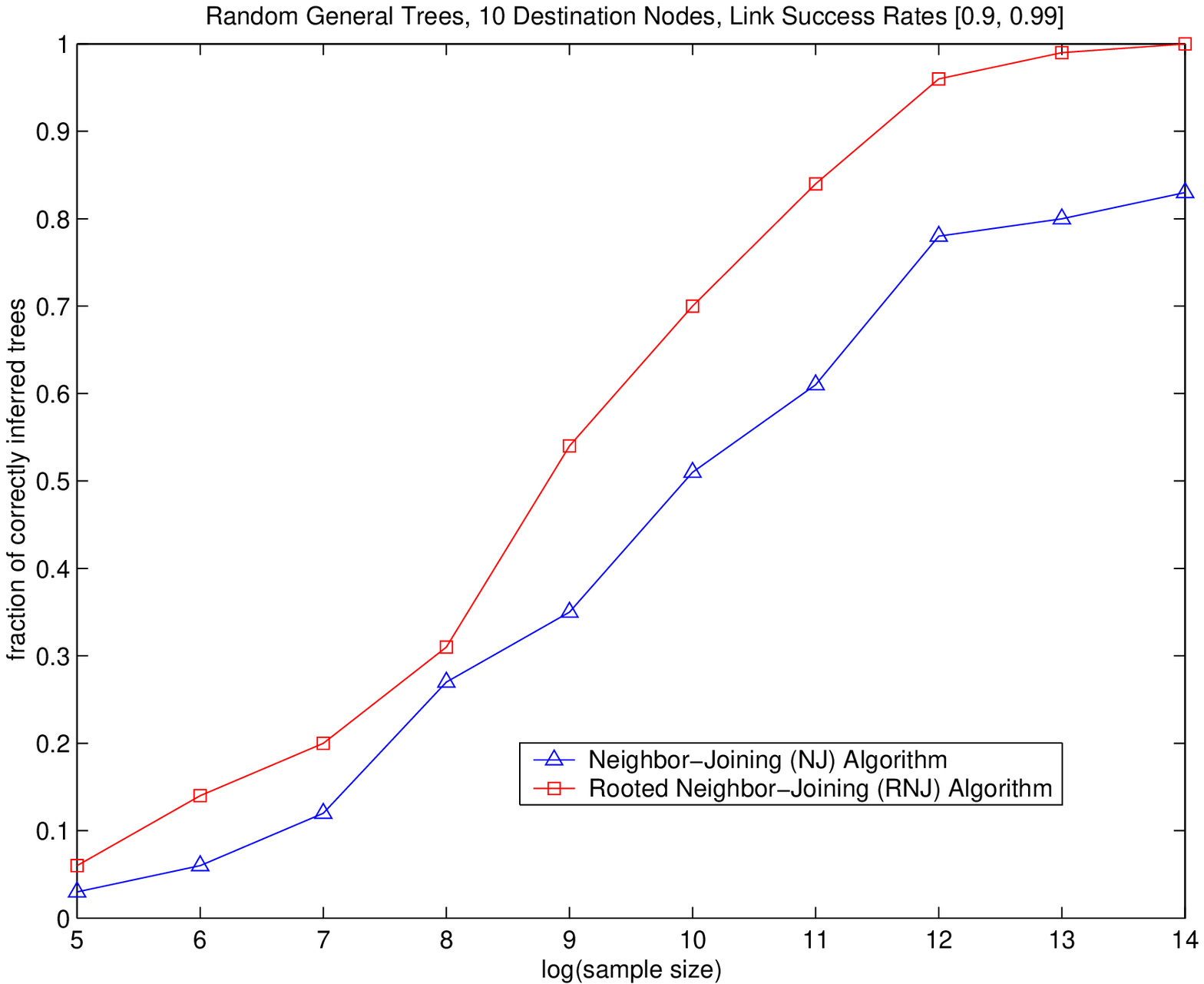, width=\columnwidth}} \caption{General trees: fraction of
correctly inferred trees.} \label{fig:MGN10a90topo}

\center{\psfig{file=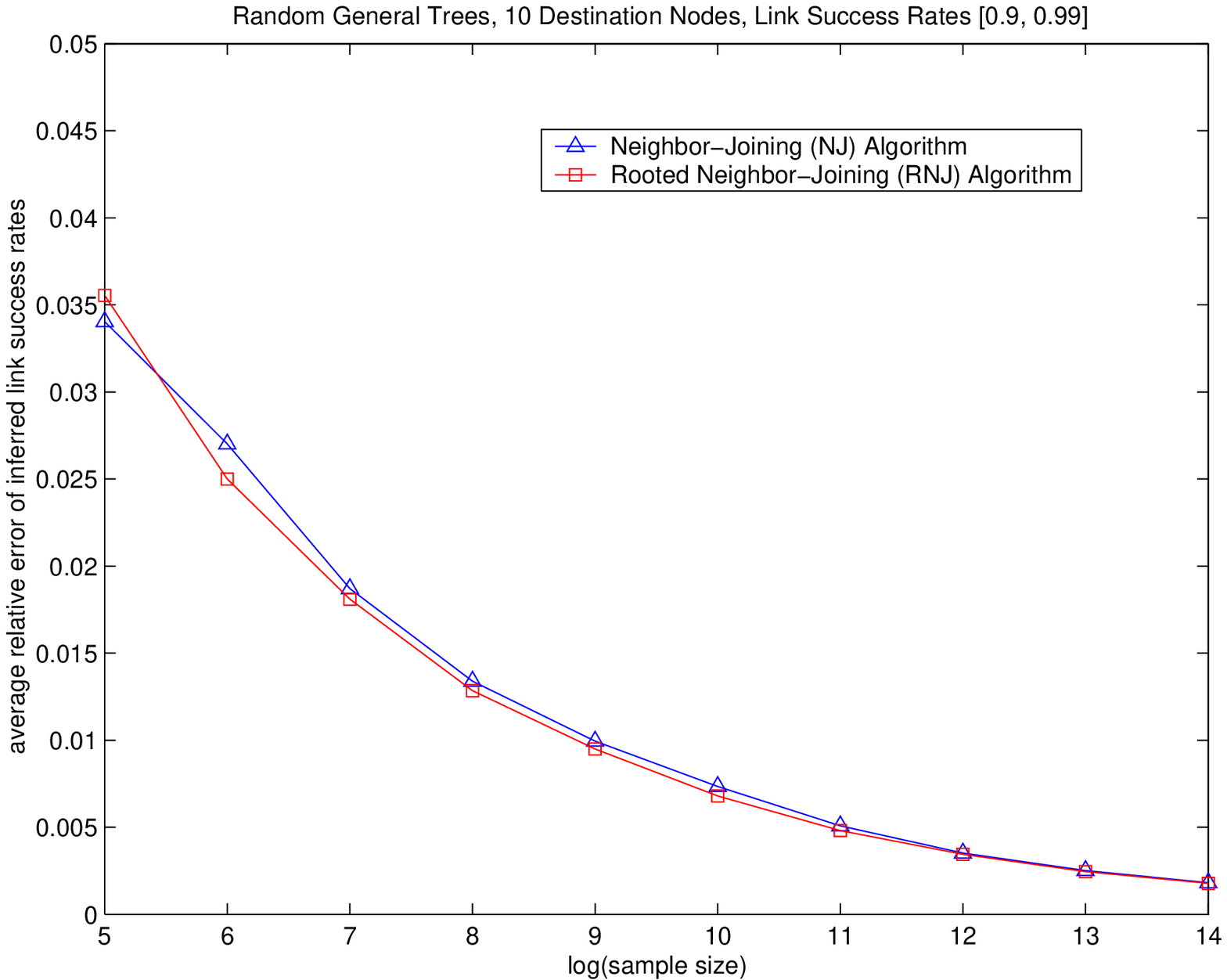, width=\columnwidth}} \caption{General trees: average
relative error of inferred link success rates.} \label{fig:MGN10a90link}
\end{figure}

\section{Simulation Evaluation}   \label{sec:NTsimulation}

In addition to analysis, we demonstrate the effectiveness of the NJ algorithm and the RNJ algorithm
via model simulation. For each experiment, we first randomly generate the tree topology and select
the link success rates in a certain range. The source node then sends a sequence of probes to the
destination nodes. The destination nodes measure the loss outcomes of the probes. We consider both
random binary trees and general trees\footnote{Like the RNJ algorithm, we extend the NJ algorithm
for general trees using a similar threshold neighbor criterion as in
(\ref{eq:thresholdcriterion}).}.

The distances between the terminal nodes are estimated from the empirical distributions of the
observed outcomes at the destination nodes as in (\ref{eq:Emetric1}). We then apply both inference
algorithms to infer the tree topology and link success rates using the estimated distances between
the terminal nodes.

We compare the inferred tree topology with the true tree topology. If the inferred topology is
correct, then we further compare the inferred link success rates $\hat{\alpha}_e$'s with the true
link success rates $\alpha_e$'s. Specifically, for each link $e$, we compute the \emph{relative
error} of the inferred link success rate (compared with the true link success rate) as follows:
\begin{eqnarray*}
\epsilon_e & = & |\frac{\hat{\alpha}_e-\alpha_e}{\alpha_e}|,
\end{eqnarray*}
and we calculate the average relative error among all links on the tree:
\begin{eqnarray*}
\epsilon_{E} & = & \frac{1}{|E|}\sum_{e\in E}\epsilon_e.
\end{eqnarray*}

Each experiment is repeated 100 times. For each inference algorithm, we compute the \emph{fraction}
of correctly inferred trees among all 100 trials (which can be viewed as the \emph{probability} of
correct topology inference of the algorithm), as well as the average value of $\epsilon_{E}$ among
the correctly inferred trees.

The results are shown in Figs. \ref{fig:MBN10a90topo}-\ref{fig:MGN10a90link}. The $x$ axis is in
log scale, \ie, it is $\log_2n$ for a sample size of $n$ probes. As we expect from our analysis,
the NJ algorithm and the RNJ algorithm are \emph{consistent}: the fraction of correctly inferred
trees of both algorithms goes to 1 (exponentially fast) as we increase the sample size, and the
average relative error of the inferred link success rates goes to 0 with an increasing sample size.

For binary trees, we observe that the NJ algorithm and the RNJ algorithm have very similar
performance; while for general trees, the RNJ algorithm has a clear advantage over the NJ algorithm
in terms of the fraction of correctly inferred trees, implying that the RNJ algorithm is more
accurate for topology inference of general trees.

We conduct experiments for trees with different sizes and ranges of link success rates, and we
observe the same pattern of the results.

\section{Multiple-Source Single-Destination Network Inference} \label{sec:NTMSSD}

In this section we study the network inference problem of estimating the routing topology and link
performance from multiple source nodes to a single destination node, in contrast to the
single-source multiple-destination network inference problem we have addressed in the previous
sections.

Again we assume that during the measurement period, the underlying routing algorithm determines a
unique path from a node to another node that is reachable from it. Therefore, the \emph{physical
routing topology} from a set of source nodes to a destination node forms a reversed directed tree.
From the physical routing topology, we can derive a \emph{logical routing tree} which consists of
the source nodes, the destination node, and the \emph{joining} nodes (internal nodes with at least
two incoming links) of the physical routing tree. Each internal node on the logical routing tree
has degree at least three, and a logical link may comprise more than one physical links. An example
is shown in Fig. \ref{fig:physical2}.

\begin{figure}[t]
\center{\psfig{file=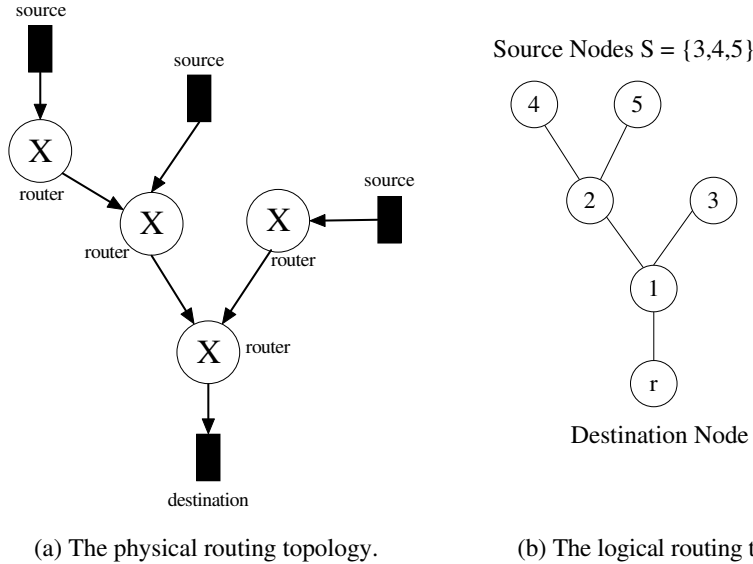, width=0.9\columnwidth}} \caption{The physical routing
topology and the associated logical routing tree with multiple source nodes and a single
destination node.} \label{fig:physical2}
\end{figure}

Let $r$ be a destination node (receiver) in the network, and $S$ be a set of source nodes that will
communicate with $r$. Let $T(S,r)=(V,E)$ denote the (logical) routing tree from nodes in $S$ to
$r$, with node set $V$ and link set $E$. Let $U=S\cup r$ be the set of terminal nodes, which are
nodes with degree one.

Each node $k\in V$ has a \emph{child} $c(k)\in V$ such that $(k,c(k)) \in E$, and a set of parents
$f(k)=\{j \in V: c(j)=k\}$, except that the destination node has no child and the source nodes have
no parents.

For notational simplification, we use $e_k$ to denote link $(k, c(k))$. Each link $e_k$ is
associated with a performance parameter $\theta_k$ (\eg, success rate, delay distribution,
utilization, etc.) that we want to estimate. The network inference problem involves using
measurements taken at the terminal nodes to infer
\begin{itemize}
\item[(1)] the topology of the (logical) routing tree;

\item[(2)] link performance parameters $\theta_{e}$ of the links on the routing tree.
\end{itemize}

\subsection{Reverse Multicast Probing} \label{sec:reversemulticast}

Similar to multicast probing from a source node to a set of destination nodes, we can have
\emph{reverse multicast probing} from a set of source nodes to a single destination node. We
illustrate the idea of reverse multicast using Fig. \ref{fig:physical2}(b) as the example.

Under a reverse multicast probing, source nodes 4 and 5 will send a packet (probe) to their child
node 2. Node 2 may receive both packets, or one of them, or none of them (because of packet loss).
If node 2 receives at least one packet from its parents, it will combine (\eg, concatenate) the
packets and sends the combined packet (as a probe) to its child node 1. Otherwise, node 2 will send
nothing. Similarly, source node 3 will send a packet to its child node 1. Node 1 combines the
packets received from its parents (if any) and sends the combined packet to the destination node
$r$. The whole process is like the reverse process of multicasting a probe from node $r$ to the
other nodes on the routing tree.

For a probe sent from the source nodes in $S$ to the destination node $r$, we define a set of
\emph{link state variables} $Z_e$ for all links on the routing tree $T(S,r)$. Using link loss
inference as the example, $Z_{e}$ is a Bernoulli random variable which takes value 1 with
probability $\alpha_e$ if the probe can go through link $e$, and takes value 0 with probability
$1-\alpha_e\stackrel{\Delta}{=}\bar{\alpha}_e$ if the probe is lost on the link.

For each node $k$ on the routing tree, we use $X_{k}$ to denote the (random) outcome of the probe
sent from node $k$ observed by the destination node $r$. For link loss inference, $X_k$ takes value
1 if $r$ successfully receives the probe sent from node $k$, and takes value 0 otherwise. It is
clear that for any source node $i$,
\begin{equation}
X_i = Z_{e_i} \cdot X_{c(i)} =\prod_{e\in \mathcal{P}(i,r)}Z_e.
\end{equation}

If $0<\alpha_e<1$ for all links, then we can construct an additive metric $d_l$ with link length
\begin{eqnarray} \label{eq:linkpara1rev}
d_l(e) & = & - \log \alpha_e, \phantom{aa} \forall e\in E.
\end{eqnarray}

For any pair of source nodes $i, j\in S$, let $\underline{ij}$ denote their \emph{nearest common
descendant} on $T(S,r)$ (\ie, the descendant of both nodes $i$ and $j$ that is closest to $i$ and
$j$ on the routing tree). For example, in Fig. \ref{fig:physical2}(b), the nearest common
descendant of source nodes 4 and 5 is node 2, and the nearest common descendant of source nodes 3
and 4 is node 1.

Under the spatial independence assumption that the link states are independent from link to link,
for any pair of source nodes $i$ and $j$, we have
\begin{eqnarray*}
\mathbb{P}(X_i=1)    & = &  \prod_{e\in\mathcal{P}(i,r)}\alpha_e, \\
\mathbb{P}(X_j=1)    & = &  \prod_{e\in\mathcal{P}(j,r)}\alpha_e, \\
\mathbb{P}(X_iX_j=1) & = &  \prod_{e\in\mathcal{P}(i,\underline{ij})}\alpha_e
\prod_{e\in\mathcal{P}(j,\underline{ij})}\alpha_e
\prod_{e\in\mathcal{P}(\underline{ij},r)}\alpha_e.
\end{eqnarray*}

Therefore, the distances between the terminal nodes, $d_l(U^2)$, can be computed by
($\mathbb{P}(X_r=1)=1$):
\begin{eqnarray}\label{eq:dmetric1rev}
d_l(i,j) = \log \frac{\mathbb{P}(X_i=1)\mathbb{P}(X_j=1)}{\mathbb{P}^2(X_iX_j=1)},
\phantom{aa}i,j\in U.
\end{eqnarray}

We can see that the mathematical model of a reverse multicast probing on a routing tree (with
multiple source nodes and a single destination node) is similar to the mathematical model of a
multicast probing on a routing tree (with a single source node and multiple destination nodes).
Therefore, the additive-metric framework can be directly applied to analyze and solve the
multiple-source single-destination network inference problem. Specifically, we can construct
additive metrics, estimate the distances between the terminal nodes from end-to-end measurements,
and apply the distance-based algorithms to infer the routing tree topology and the link performance
metrics.

\subsection{Passive Network Monitoring in Wireless Sensor Networks}

Although the current Internet does not support reverse multicast probing because internal nodes
(routers) do not combine packets sent from different source nodes to a destination node, reverse
multicast can be deployed in wireless networks (\eg, \cite{DirectedDiffusion},
\cite{SensorFactorGraph}) for efficiently data collecting.

A typical scenario in wireless sensor networks for data collecting is as follows. A base station (a
receiver) will first propagate an \emph{interest message} into the network via flooding or
constrained/directional flooding. An interest message could be a query message which specifies what
the base station wants (\eg, temperature statistics). A node, when first receives the interest
message from another node, will set that node as its \emph{child} and forward the interest message
to its own neighbors excluding its child. Hence the interest propagation procedure serves both to
disseminate the interest message, and to set up a \emph{reverse path} from each node to the base
station.

When a sensor node which has the data of interest (a source node) receives the interest message, it
can send the data back to the base station using the reverse path (\ie, it sends the data to its
child). Assume each source node has a unique ID (\eg, the geographical location of the node). The
data sent by a source node to the base station also include the source node's ID so the base
station knows from where it receives the data.

If each node selects only one node as its child, \ie, if there is a unique path from a node to the
base station, then we know that the routing topology (undirected version) from the source nodes to
the base station is a tree. We call it a \emph{data collecting tree}. Each internal node on the
tree only needs to maintain the information of a set of parents that it will receive data from, and
a child that it will send data to.

Suppose \emph{directed diffusion} \cite{DirectedDiffusion} is applied on the data collecting tree,
under which an internal node will aggregate (\eg, combine, compress, code, etc.) the data sent from
its parents and then send the aggregated data to its child. Then this process is like a reverse
multicast probing process as we described in Section \ref{sec:reversemulticast}. Using the
algorithms we have developed in this paper, the base station can infer: (1) the topology of the
data collecting tree; (2) the link performance (\eg, packet delivery rate) of every link on the
data collecting tree.

There are several advantages for the base station to do network inference based on the collected
data from the sensor nodes. First, the (internal) sensor nodes do not need to measure and infer the
link performance which can save their resources; while normally the base station has sufficient
battery and computation power so it is competent for the network inference task. Second, this is a
\emph{passive network monitoring framework} so no extra probing traffic is generated. In addition,
since the inference is based on real data transmission, the inferred link performance metrics are
more accurate and meaningful.

\section{Conclusion}   \label{sec:NTsummary}

In this paper we address the network inference problem of estimating the routing topology and link
performance in a communication network. We propose a new, general framework for designing and
analyzing network inference algorithms based on additive metrics using ideas and tools from
phylogenetic inference. The framework is applicable to a variety of measurement techniques. Based
on the framework we introduce and develop several distance-based inference algorithms. We provide
sufficient conditions for the correctness of the algorithms. We show that the algorithms are
computationally efficient (polynomial-time), consistent (return correct topology and link
performance with an increasing sample size), and robust (can tolerate a certain level of
measurement errors). In addition, we establish certain optimality properties of the algorithms
(\ie, they achieve the optimal $l_\infty$-radius) and demonstrate their effectiveness via model
simulation. The framework provides powerful tools that enable us to infer and estimate the
structure and dynamics of large-scale communication networks.

\end{document}